\documentclass[letterpaper, 10 pt, conference]{ieeeconf}  

\IEEEoverridecommandlockouts                              
\usepackage{mathtools}
\usepackage{relsize}
\usepackage{tikz}
\graphicspath{{./Figures/}} 

\usepackage[margin=0.8cm]{caption}
\newcommand{\norm}[1]{\left\lVert#1\right\rVert}
\usepackage{pgfplots}
\pgfplotsset{compat=newest}
\usepackage{tabularx}
\pgfplotsset{plot coordinates/math parser=false}
\usepackage{algorithm}
\usepackage{soul}
\usepackage{dsfont}
\usepackage{tabto}

\newcommand{\Pc}{\mathcal{P}}
\newcommand{\Ec}{\mathcal{E}}

\newtheorem{theorem}{Theorem}
\newtheorem{assumption}{Assumption}
\newtheorem{corollary}{Corollary}

\newtheorem{remark}{Remark}

\newtheorem{proposition}{Proposition}

\usepackage{lineno,amssymb,subcaption,algpseudocode}
\newcommand\scalemath[2]{\scalebox{#1}{\mbox{\ensuremath{\displaystyle #2}}}}

\newcommand{\B}{B}

\newcommand{\R}{\mathbb{R}}
\newcommand{\D}{\mathbb{D}}
\newcommand{\So}{\mathbb{S}}

\newcommand{\1}{\mathbf{1}}
\newcommand{\0}{\mathbf{0}}
\newcommand{\I}{\mathbf{I}}
\usepackage{accents}

\usepackage{cuted}
\usepackage{amsmath,amssymb,amsfonts}
\usepackage{graphicx}
\usepackage{paralist}

\usepackage{verbatim}
\usepackage[normalem]{ulem}

\usepackage{bm}
\usepackage{tikz}
\usetikzlibrary{calc,patterns,decorations.pathmorphing,decorations.markings}
\usepackage{amsmath} 
\usepackage{amssymb}  
\usepackage{mathtools}
\usepackage{dsfont}
\usepackage{bm}
\usepackage{xcolor}
\definecolor{blue_set}{RGB}{204,229.5,255}
\definecolor{pink_set}{RGB}{255,204,229.5}
\definecolor{grey_set}{RGB}{153,153,153}
\definecolor{green_set}{RGB}{102,204,102}
\definecolor{cyan_set}{RGB}{69,243,248}
\definecolor{yellow_set}{RGB}{229.5000,229.5000,114.7500}
\definecolor{red_set}{RGB}{231,172,116}
\definecolor{red_border}{RGB}{229.5,114.75,114.75}
\definecolor{terminal_set}{RGB}{102,204,255}
\definecolor{mRPI_set}{RGB}{102,255,102}
\usepackage{filecontents}
\usetikzlibrary{positioning}
\usetikzlibrary{calc}
\usetikzlibrary{shapes,snakes}

\newcommand\munderbar[1]{\underaccent{\bar}{#1}}

\usepackage{url}

\definecolor{wheat}{rgb}{0.96,0.87,0.70}
\definecolor{mario}{rgb}{0.8,0.8,1}
\definecolor{SamComm}{rgb}{0.9,0.9,0.1}
\definecolor{ao}{rgb}{0.0, 0.5, 0.0}

\title{\LARGE \bf
	Data-driven synthesis of Robust Invariant Sets and Controllers
}

\author{Sampath Kumar Mulagaleti, Alberto Bemporad, Mario Zanon
	\thanks{The authors are with IMT School  for  Advanced  Studies  Lucca, Italy.}%
	\thanks{\texttt{\scriptsize\{s.mulagaleti,alberto.bemporad,mario.zanon\}@imtlucca.it}}%
}

\begin{document}

	\maketitle
	\thispagestyle{empty}
	\pagestyle{empty}

	\begin{abstract}
		This paper presents a method to identify an uncertain linear time-invariant (LTI) prediction model for tube-based Robust Model Predictive Control (RMPC). The uncertain model is determined 
		from a given state-input dataset by formulating and solving a Semidefinite Programming problem (SDP), that also determines a static linear feedback gain and corresponding invariant sets satisfying the inclusions required to guarantee recursive feasibility and stability of the RMPC scheme, while minimizing an identification criterion. 
		%
		As demonstrated through an example, the proposed
			concurrent approach provides less conservative invariant sets than a sequential approach.

	\end{abstract}
	
	\section{Introduction}
	The tube-based RMPC scheme of \cite{Mayne2005} is a popular method to design robust feedback controllers using
	LTI plant models
	\begin{equation}
		\label{eq:system}
		\scalemath{0.95}{
			x(t+1)={A}x(t)+{B}u(t)+w(t),}
	\end{equation}
	subject to state constraints $x \in \mathcal{X} \subset \R^{n_x}$, input constraints $u \in \mathcal{U} \subset \R^{n_u}$, and additive unknown but bounded disturbances $w \in \mathcal{W}$.
	Besides model~\eqref{eq:system}, the RMPC scheme also requires a feedback gain $K$, and a Robust Positive Invariant (RPI) set \cite{Blanchini2015} in which the state of the  system $x(t+1)=({A+BK})x(t)+w(t)$ can be enforced to persistently belong.

	Given a model $(A,B,\mathcal{W})$, the problem of computing constrained RPI sets has been well studied in the literature. We focus on polytopic RPI sets for their reduced conservatism \cite{Blanchini2015}. 
	%
	%
	%
	%
	Given $K$, methods to compute tight invariant approximations of the minimal RPI (mRPI) sets were presented in 
	\cite{Rakovic2005,Rakovic2013,Trodden2016}. They characterize the uncertainty tube that bounds the deviation of the actual state trajectory from a central nominal one.
	%
	%
	Similarly, 
	methods to compute maximal PI (MPI) sets were presented in \cite{Gilbert1991,Kolmanovsky1998}.
	They can be used as terminal sets in RMPC to guarantee feasibility and stability.
	%
	It is known that improved RPI sets can be computed by also optimizing over $K$. 
	%
	%
	%
	In \cite{Tahir2015,Liu2019}, methods to compute RPI sets along with $K$ were presented.
	Earlier approaches in~\cite{Kothare1996,Bem98clw} optimize over $K$ and reduce conservativeness in RMPC.
	
	In order to identify a model $(A,B,\mathcal{W})$,
	physics-based, regression approaches and/or set-membership approaches \cite{Ljung1986,Kosut1992} can be used.
	%
	Methods that take control design into account while performing system identification were presented in \cite{Helmicki1991,Terzi2019}. It was demonstrated in \cite{Yuxiao2021} that if
	system identification can be combined with RPI set computation, then conservativeness in the computed RPI sets can be reduced. Motivated by this observation, we present a method to concurrently select a model $(A,B,\mathcal{W})$ and synthesize RPI sets for RMPC.
		%
		Alternative methods that directly compute feedback controllers using an implicit plant description based on measured trajectories were presented in \cite{DePersis2020,Berberich2021,Berberich2020,coulson2019dataenabled,vanwaarde2020noisy}.
		However, these methods cannot be used directly to select a model and RPI sets for RMPC synthesis. While a similar trajectory-based method in \cite{bisoffi2020controller} can be used to synthesize RPI-inducing controllers in a given polyhedral set, it cannot be used directly for RPI set synthesis.
%
		
	%
	%
	%
	%
	
	{\textit{Contribution:}} We consider a dataset of state-input measurements from a plant, and present a method to identify an LTI model \eqref{eq:system}, along with RPI sets suitable for RMPC synthesis. To this end, we characterize a set of models $(A,B,\mathcal{W})$ that can describe the plant behavior, and use nonlinear matrix inequality (NLMI)-based results from \cite{Liu2019} on RPI set computation to formulate a NonLinear Program with Matrix Inequalities (NLPMI) that selects a model \eqref{eq:system} along with suitable RPI sets and a corresponding feedback matrix $K$.
	%
	%
	We then present a method to solve the NLPMI based on a Sequential Convex Programming (SCP) approach that 
	we tailor to preserve
	feasibility of the iterates and satisfy a cost decrease condition.
	Finally, we demonstrate the efficacy of the method using a simple numerical example.
	%
	%
	
	\textit{Notation:} $\Pc(A,b):=\{x\in\R^n:\ -b \leq Ax \leq b\}$ is a symmetric polytope, and $\Ec(Q,r):=\{x\in\R^n:\ x^{\top}Q x \leq r\}$ is an ellipsoid. The set of $m$ dimensional positive vectors is denoted as $\R^m_+$, positive definite $m \times m$ diagonal matrices is denoted as $\D_+^m$, positive definite $m \times m$ symmetric matrices as $\So_+^m$. The symbols $\1$, $\0$ and $\I$ denote all-ones, all-zeros, and identity matrix, respectively. The set $\mathbb{I}_m^n:=\{m,\cdots,n\}$ is the set of natural numbers between $m$ and $n$. $T_i$ and $T_{ij}$ denote row $i$ and element $(i,j)$ of matrix $T \in \R^{n \times m}$, and $\norm{T}_{\infty}:=\max_{i \in \mathbb{I}_1^n} \sum_{j=1}^{m} |T_{ij}|$ is the $\infty$-norm of the matrix.
	We define $\norm{v}_S^2:=v^{\top}Sv$, and use $*$ to represent symmetrically identifiable matrix entries. Given two compact sets $\mathcal{S}_1,\mathcal{S}_2 \subset \R^{n}$, the Minkowski sum is defined as $\mathcal{S}_1 \oplus \mathcal{S}_2:=\{x+y:x\in\mathcal{S}_1,y\in\mathcal{S}_2\}$, and the Minkowski difference as $\mathcal{S}_1 \ominus \mathcal{S}_2:=\{x:\{x\}\oplus\mathcal{S}_2\subseteq\mathcal{S}_1\}$. We write
		$C_1 \mathcal{P}(A_1,b_1) \oplus C_2 \mathcal{P}(A_2,b_2)=[C_1 \ C_2] \mathcal{P}(\mathrm{diag}(A_1,A_2),[b_1^{\top} b_2^{\top}]^{\top})$, where	 $\mathrm{diag}(A_1,A_2):=\begin{bmatrix}A_1 & \0 \\ \0 & A_2 \end{bmatrix}$ is a block-diagonal matrix.
	%
	\section{Problem formulation}
	We briefly recall the tube-based RMPC scheme from \cite{Mayne2005}.
	%
	Given system \eqref{eq:system}, consider the nominal model $\hat{x}(t+1)={A}\hat{x}(t)+{B}\hat{u}(t)$,
	and parameterize the plant input as $u(t)=\hat{u}(t)+{K}(x(t)-\hat{x}(t)),$
	where ${K}$ is a static feedback gain. Assuming that the feedback gain is stabilizing for $(A,B)$,
	the state error $\Delta x :=x - \hat{x}$ with dynamics $\Delta x(t+1)=({A}+{B}{K})\Delta x(t)+w(t)$ belongs to the RPI set $\Delta \mathcal{X}$
	\begin{align}
		\label{eq:RPI_property}
		\scalemath{0.95}{
			\text{if $\Delta x(0) \in {\Delta \mathcal{X}}$, \quad and } \quad 
			(A+BK)\Delta \mathcal{X}\oplus \mathcal{W} \subseteq \Delta \mathcal{X}}.
	\end{align}
	%
	Hence, $x$ always belongs to the uncertainty tube with cross-section $\Delta \mathcal{X}$ around $\hat{x}$, i.e.,
	%
	%
	$x(t) \in \hat{x}(t)\oplus {\Delta \mathcal{X}}, \forall t \geq 0$.
	The RMPC scheme then enforces $\hat{x} \in \mathcal{X} \ominus {\Delta \mathcal{X}}$ and $\hat{u} \in \mathcal{U} \ominus {K}{\Delta \mathcal{X}}$, and computes
	$\mathbf{z}:=\{\hat{x}(t),\ldots,\hat x(t+N),\hat{u}(t),\ldots,\hat u(t+N-1)\}$ online given $x(t)$ by solving
	%
	%
	%
	\begin{align}
		\scalemath{1}{\min_{\mathbf{z}} \ } & \scalemath{0.9}{\sum\limits_{s=t}^{t+N-1} 
			\norm{\begin{bmatrix}\hat{x}(s)^{\top} \quad \hat{u}(s)^{\top}\end{bmatrix}^{\top}}^2_{
				{H}_{\mathrm{Q}}}
			+
			\norm{\hat{x}(t+N)}^2_{{P}_{\mathrm{Q}}}}
		&& \hspace{2em} \nonumber \\
		\scalemath{0.9}{	\mathrm{s.t.} }\ \ 
		& \scalemath{0.95}{\hat{x}(s+1)={A} \hat{x}(s)+{B} \hat{u} (s), } &&\hspace{-30pt} \scalemath{0.9}{ s \in \mathbb{I}_t^{t+N-1},  }
		\nonumber \\
		& \scalemath{0.9}{ \hat{x}(s) \in \mathcal{X} \ominus {\Delta \mathcal{X}}, \  \hat{u}(s) \in \mathcal{U} \ominus {K\Delta \mathcal{X}}, }&&\hspace{-30pt}   \scalemath{0.9}{ s \in \mathbb{I}_{t+1}^{t+N-1},  } \nonumber \nonumber \\
		& \scalemath{0.9}{x(t) \in \left\{\hat{x}(t)\right\} \oplus {\Delta \mathcal{X}}, \ \ \hat{x}(t+N) \in {\mathcal{X}_{\mathrm{t}}}},
		\hspace{-30em} && \label{eq:RMPC_controller}
	\end{align}
	where $\mathcal{X}_{\mathrm{t}}$ is the terminal set.
	We assume that the set $\Delta \mathcal{X}$ is \textit{small} enough for feasibility of Problem \eqref{eq:RMPC_controller}, i.e.,
		\begin{align}
		\label{eq:smallRPI}
		\scalemath{0.9}{
			\Delta \mathcal{X} \subset \mathcal{X}, \ K\Delta \mathcal{X} \subset \mathcal{U},}
		\end{align}
	and $\mathcal{X}_{\mathrm{t}}$ is a PI set for $\hat{x}(t+1)=(A+BK)\hat{x}(t)$ that satisfies
	\begin{align}
			\label{eq:PI_required}
		\scalemath{0.9}{ 
			({A}+{B}{K}){\mathcal{X}_{\mathrm{t}}}\subseteq {\mathcal{X}_{\mathrm{t}}}\subseteq \mathcal{X}\ominus  {\Delta \mathcal{X}}, \quad {K}{\mathcal{X}_{\mathrm{t}}} \subseteq \mathcal{U} \ominus {K} {\Delta \mathcal{X}}.} 	
	\end{align}
%
%
	Then $\Omega_{N}:=\{x: \text{\eqref{eq:RMPC_controller} is feasible with $x(t)=x$}\}$
	is such that for each $x(t) \in \Omega_N$, there recursively exists an optimal solution $\mathbf{z}_*:=\{\hat{x}_*(t),\ldots,\hat{x}_*(t+N),\hat{u}_*(t),\ldots,\hat{u}_*(t+N-1)\}$ \cite[Proposition 2]{Mayne2005}. Then, the input 
	$u(t):=\hat{u}_*(t)+{K}(x(t)-\hat{x}_*(t))$ is applied to the plant. Moreover, if $(H_{\mathrm{Q}},P_{\mathrm{Q}})$ are such that $P_{\mathrm{Q}}$ is the solution of the Discrete Algebraic Riccati equation (DARE) formulated using $H_{\mathrm{Q}}$ for the system $(A,B)$, and $K$ is corresponding optimal feedback gain, $\Delta \mathcal{X}$ is exponentially stable from every $x \in \Omega_N$ \cite[Theorem 1]{Mayne2005}.
	%
	%
	%
	
 \textit{\textbf{Problem description}}: 
		We consider a plant with dynamics $x(t+1)=f_{\mathrm{tr}}(x(t),u(t),v(t))$ that is subject to bounded inputs $u(t)$ and unknown but bounded disturbances $v(t) \in \mathcal{V}_{\mathrm{tr}}\subset \R^{n_v}$. Assuming that $f_{\mathrm{tr}}$ and $\mathcal{V}_{\mathrm{tr}}$ are unknown, we collect a dataset $\mathcal{D}:=\{x_{\mathrm{D}}(t),u_{\mathrm{D}}(t), \ t \in \mathbb{I}_1^{T}\}$ of state-input measurements from the plant. Using $\mathcal{D}$, we propose a method to compute $(A,B,\mathcal{W},K,\Delta \mathcal{X},\mathcal{X}_{\mathrm{t}})$ that satisfy \eqref{eq:system}, \eqref{eq:RPI_property}, \eqref{eq:smallRPI}, \eqref{eq:PI_required} required for RMPC synthesis, while optimizing some criterion.
	%
	%
	%
		In the sequel, we assume that the set $\Phi_{*}$ of all possible vectors $[x^{\top} \  u^{\top}]^{\top}$ that can be collected from the plant is a bounded. This is a natural consequence if $f_{\mathrm{tr}}$ is open-loop stable.
%
		Then, $\Phi_T:=\{[x_{\mathrm{D}}(t)^{\top} \ u_{\mathrm{D}}(t))^{\top}]^{\top}, t \in \mathbb{I}_1^T\} \subset \Phi_{*}$. 
		 Then, we define $	\scalemath{0.9}{
			\mathcal{J}_{*}:=\begin{Bmatrix}\begin{bmatrix}x \\ u \\ x_{+}\end{bmatrix} : \begin{matrix} x_{+} = f_{\mathrm{tr}}(x,u,v), \vspace{5pt} \\  \forall \ [x^{\top} \ u^{\top}]^{\top} \in \Phi_{*}, \ \forall v \in \mathcal{V}_{\mathrm{tr}}\end{matrix} \end{Bmatrix}.}$
		Since $\Phi_{*}$, $\mathcal{V}_{\mathrm{tr}}$ and $u$ are bounded, $\mathcal{J}_{*}$ is also a bounded set.
		Finally, we denote the measured set $\mathcal{J}_{T}:=\{[x_{\mathrm{D}}(t)^{\top}\ u_{\mathrm{D}}(t)^{\top} \ x_{\mathrm{D}}(t+1)^{\top}]^{\top}, \ t \in \mathbb{I}_1^{T-1}\} \subset \mathcal{J}_*$.		
	\begin{remark}
		\label{rem:FB_performance}
		Given $(A,B,K)$, matrices $(H_{\mathrm{Q}},P_{\mathrm{Q}})$ satisfying the DARE can be computed using the procedure in \cite{zanon2020constrained}. Hence, our approach involves tuning the MPC scheme. 
		Combining our approach with other MPC tuning methods such as \cite{Masti2022} is a subject of future research.	
		$\hfill\square$
	\end{remark}
	\vspace{-4pt}
	\section{Identification based on invariant sets}
	\vspace{-2pt}
	\label{sec:main_section}
	To compute $(A,B,\mathcal{W},K,\Delta \mathcal{X},\mathcal{X}_{\mathrm{t}})$ using the dataset $\mathcal{D}$ by optimization,  
	we first characterize the set of models $(A,B,\mathcal{W})$ that are suitable to model the underlying plant.
	%
			
		1) \textbf{\textit{ Characterization of feasible models}}: 
		We consider a disturbance set parametrized as $\mathcal{W}:=\mathcal{P}(F,d), d \in \R_+^{m_w}$. We assume for simplicity that $F$ is fixed a priori.
		Then, system \eqref{eq:system} with $w \in \mathcal{P}(F,d)$ is suitable for RMPC synthesis if it can model all possible state transitions of the plant as
		%
		\begin{align}
			\label{eq:robust_prediction_inclusion}
			\scalemath{0.97}{
				\begin{matrix}
					\begin{matrix}f_{\mathrm{tr}}(x,u,v) \in \{Ax+Bu\}\oplus \mathcal{P}(F,d), \vspace{2pt} \\
						\forall \ [x^{\top} \ u^{\top}]^{\top} \in \Phi_{*}, \quad \forall \ v \in \mathcal{V}_{\mathrm{tr}}. \end{matrix} 
				\end{matrix}
			}
		\end{align}	
		Defining the prediction error $\scalemath{0.98}{\zeta(A,B,z):=x_+-Ax-Bu}$ with $\scalemath{0.98}{z:=[x^{\top} \ u^{\top} \ x_+^{\top}]^{\top}}$, \eqref{eq:robust_prediction_inclusion} holds for a given $(A,B)$ if and only if $\zeta(A,B,z) \in \mathcal{P}(F,d), \ \forall  \ z \in \mathcal{J}_*$
		by definition of $\mathcal{J}_*$. Hence,
		$\Sigma_*:=\{(A,B,d):\zeta(A,B,z) \in \mathcal{P}(F,d), \ \forall \ z \in \mathcal{J}_*\}$ characterizes the set of all models $(A,B,d)$ satisfying \eqref{eq:robust_prediction_inclusion}.
		%
		However, constructing $\Sigma_*$ is not possible since we only have access to the measured set $\mathcal{J}_T$. 
		Hence, we characterize
			\begin{align*}
			\scalemath{0.93}{
				\Sigma_T(\theta_T):=\begin{Bmatrix}{\begin{pmatrix}A, \\B, \\ d\end{pmatrix}: \begin{matrix}\zeta(A,B,z) \in \mathcal{P}(F,d-\kappa_T(A,B)\1)\vspace{2pt} \\
							\scalemath{0.93}{\kappa_T(A,B)=\norm{F[-A \ -B \ \I]}_{\infty}\theta_T}, \vspace{2pt} \\
							d>\kappa_T(A,B)\1,\ \forall \ z \in \mathcal{J}_{T}  \end{matrix} }\end{Bmatrix}}
		\end{align*}
	  using $\mathcal{J}_T$, where $\theta_T:=\mathbf{d}_{\infty}(\mathcal{J}_{*},\mathcal{J}_T)$ is the Hausdorff distance between the sets $\mathcal{J}_T$ and $\mathcal{J}_*$ in $\infty$-norm, and is given by $\mathbf{d}_{\infty}(\mathcal{J}_*,\mathcal{J}_{{T}}):=\max_{z_*\in \mathcal{J}_*} \min_{z\in \mathcal{J}_{{T}}} \norm{z-z_*}_{\infty}$ since the inclusion $\mathcal{J}_T \subset \mathcal{J}_{\infty}$ holds for every $T>0$.
	  	\begin{assumption}
	  	\label{ass:sufficient_excitation_assumption}
	  	($a$) $\Sigma_T(\theta_T) \neq \emptyset$; ($b$)
	  	$\forall \ \theta \in \R^1_+, \exists \ \tilde{T}<\infty$ such that $\mathbf{d}_{\infty}(\mathcal{J}_*,\mathcal{J}_{\tilde{T}})\leq \theta$.
	  \end{assumption}
%
%
%
		%
		%
		Assumption \ref{ass:sufficient_excitation_assumption} implies that as $T \to \infty$, the set $\mathcal{J}_*$ is densely covered by $\mathcal{J}_T$: this is an assumption on the persistence of excitation of inputs, and bound-exploring property of the disturbances acting on the underlying plant \cite[Section 3.2]{Terzi2019}.
		%
		%
		%
		%
		\begin{theorem}
			\label{thm:robust_model_theorem}
			If Assumption 1 holds, then \eqref{eq:robust_prediction_inclusion} holds for all models $\scalemath{1}{(A,B,d) \in \Sigma_T(\theta_T)}$. $ \hfill\square$
%
			\vspace{1pt}
		\end{theorem}
		\begin{proof}
			We show that $\zeta(A,B,z) \in \mathcal{P}(F,d), \forall z \in \mathcal{J}_* , \forall (A,B,d) \in \Sigma_T(\theta_T)$.  For any $(A,B,d) \in \Sigma_T(\theta_T)$, clearly $\zeta(A,B,z) \in \mathcal{P}(F,d-\kappa_T(A,B)\1) \subset \mathcal{P}(F,d), \forall z \in \mathcal{J}_T$. By definition of the Hausdorff distance, for every remaining $\bar{z} \in \mathcal{J}_{*} \setminus \mathcal{J}_T:=\{\tilde{z}:\tilde{z} \in \mathcal{J}_{*}, \tilde{z} \notin \mathcal{J}_T\}$, $\exists \ z \in \mathcal{J}_T$ such that $||z-\bar{z}||_{\infty} \leq \theta_T$. Then, for any $(A,B,d)\in \Sigma_T(\theta_T)$,
			\begin{align*}
				\scalemath{0.9}{
					\hspace{8pt}
					\begin{matrix}
						\hspace{-40pt}F\zeta(A,B,\bar{z})= F(\zeta(A,B,\bar{z})-\zeta(A,B,{z})+\zeta(A,B,{z}))\vspace{1pt}\\
						\hspace{14pt}\leq
						F[-A \ -B \ \I](\bar{z}-{z})+d-\kappa_T(A,B)\1 \vspace{1pt}\\ 
						\hspace{41pt}
						\leq
						||F[-A \ -B \ \I]||_{\infty}\theta_T\1+d-\kappa_T(A,B)\1  = d,
				\end{matrix}}
			\end{align*}
			where the second step follows from definition of $\Sigma_T(\theta_T)$, 
			and third step from the definition of $\infty$-norm, the Cauchy-Schwarz inequality and $||\bar{z}-{z}||_{\infty} \leq \theta_T$. Using similar arguments, the condition $-d \leq F\zeta(A,B,\bar{z})$ follows, thus concluding that $\zeta(A,B,\bar{z}) \in \mathcal{P}(F,d), \forall \ \bar{z} \in \mathcal{J}_*\setminus \mathcal{J}_T$.
		\end{proof}	
		%
		
		Theorem \ref{thm:robust_model_theorem}
		implies that every $(A,B,d) \in \Sigma_T(\theta_T)\subset \Sigma_*$ is a feasible model for RMPC synthesis. However, $\Sigma_T(\theta_T)$ cannot be constructed from data since $\theta_T$ is unknown.
		To tackle this issue, we follow the standard approach of inflating the disturbance set using some parameter (e.g., \cite{Terzi2019}):
		%
		we propose to select some $\hat{\theta}_T>0$, and approximate 
		$\Sigma_T(\theta_T)$ with $\hat{\Sigma}_T:=\Sigma_T(\hat{\theta}_T)$ under the following assumption.
		\begin{assumption}
			\label{ass:K_hat_ass}
			$\hat{\theta}_T \geq \theta_T=\mathbf{d}_{\infty}(\mathcal{J}_{*},\mathcal{J}_T).$ 
			$\hfill\square$
		\end{assumption}
				Under Assumption \ref{ass:K_hat_ass}, we have $\hat{\Sigma}_T \subseteq \Sigma_T({\theta}_T)$. Hence, every $(A,B,d) \in \hat{\Sigma}_T$ is suitable for RMPC synthesis. In the sequel, we assume that a $\hat{\theta}_T$ satisfying Assumption \ref{ass:K_hat_ass} is selected. We then encode $\hat{\Sigma}_T$ as the set of linear constraints
				\begin{align*}
					\scalemath{0.95}{
						\hat{\Sigma}_T=\begin{Bmatrix}{\begin{pmatrix}A, \\ B, \\ d\end{pmatrix}: \begin{matrix} \zeta(A,B,z) \in \mathcal{P}(F,d-\lambda \hat{\theta}_T\1), d >\lambda \hat{\theta}_T\1,  \vspace{2pt} \\
									-\mathcal{Z}\leq F[-A \ -B \ \I]\leq \mathcal{Z},\mathcal{Z} \geq \0, \\	{\scalemath{0.97}{\Sigma_{j=1}^{2n_x+n_u}\mathcal{Z}_{ij}}} \leq \lambda, \forall \ i \in \mathbb{I}_1^{m_w}, 
									\forall \ z \in \mathcal{J}_{T}  \end{matrix} }\end{Bmatrix}}
				\end{align*} 
				using the definition of $\infty$-norm for matrices, where $\mathcal{Z} \in \R_+^{m_w \times (2n_x+n_u)}$ is a slack variable matrix.
		%
		%
		%
		We reiterate that since Assumption \ref{ass:K_hat_ass} cannot be verified directly using data, robustness guarantees with respect to the underlying plant can only be provided in theory. However, if Assumption~\ref{ass:sufficient_excitation_assumption}($b$) holds, the distance $\theta_T \to 0$ for large $T$.
		Hence, guessing some $\hat{\theta}_T \approx 0$
		can satisfy Assumption \ref{ass:K_hat_ass} for large datasets. 
		Moreover, the validity of a given $\hat{\theta}_T$ can be checked by verifying the existence of a model $(A,B,d) \in \hat{\Sigma}_T$ explaining a validation dataset. On the other hand, computation of a $\hat{\theta}_T$ satisfying Assumption \ref{ass:K_hat_ass} is a fundamental issue in data-driven methods: while statistical techniques such as, e.g., bootstrapping can be used, the development of such methods is a future research subject.
		\begin{remark}
			($a$) In \cite{Terzi2019}, an optimal LTI model set is first computed, from which a model is selected and then feedback controllers are synthesized. 
			%
			We combine all three phases in the current work;
			%
			%
			%
			($b$) In \cite{Berberich2020}, the closed-loop dynamics of an unknown LTI plant with a known disturbance set is characterized in terms of the measured dataset, and parametrized by unknown but bounded disturbance sequences. Then, a controller is synthesized for all feasible LTI models. We instead use a model-dependent disturbance set.
			While the assumption of a known disturbance set is as strict as Assumption \ref{ass:K_hat_ass}, comparison with \cite{Berberich2020} is a subject of future research. 
			%
			%
			%
			%
			%
			$\hfill\square$
		\end{remark}

		2)\textbf{ \textit{Robust PI set design:}} We will now compute a feedback gain $K$ and corresponding invariant sets $\Delta \mathcal{X}$ and $\mathcal{X}_{\mathrm{t}}$ for some $(A,B,d) \in \hat{\Sigma}_T$. To this end, we parametrize the RPI set
	as $\Delta \mathcal{X}=\mathcal{P}(\munderbar{P},\munderbar{b}), \munderbar{b} \in \R_+^{\munderbar{m}}$, the PI terminal set as $\mathcal{X}_{\mathrm{t}}=\mathcal{P}(\bar{P},\bar{b}), \bar{b} \in \R_+^{\bar{m}}$, and
	assume that the constraint sets are $\mathcal{X}=\mathcal{P}(V^x,v^x), v^x \in \R^{m_x}_+,$ and $\mathcal{U}=\mathcal{P}(V^u,v^u), v^u \in \R^{m_u}_+$.
	Then, for some $(A,B,d) \in \hat{\Sigma}_T$, if $(K,\munderbar{P},\munderbar{b},\bar{P},\bar{b})$ satisfies
	%
	\vspace{-3pt}
	\begin{subequations}
		\label{eq:main_constraints}
		\begin{align}
			\scalemath{0.95}{
				(A+BK)\mathcal{P}(\munderbar{P},\munderbar{b})\oplus \mathcal{P}(F,d) \subseteq \mathcal{P}(\munderbar{P},\munderbar{b}), \label{eq:RPI_inclusion_p}} \\
			\scalemath{0.95}{
				(A+BK)\mathcal{P}(\bar{P},\bar{b}) \subseteq \mathcal{P}(\bar{P},\bar{b}), \label{eq:MPI_inclusion_p} }\\
			\scalemath{0.95}{	\mathcal{P}(\munderbar{P},\munderbar{b}) \oplus \mathcal{P}(\bar{P},\bar{b}) \subseteq \mathcal{X}, \label{eq:state_inclusion_p} }\\
			\scalemath{0.95}{	K\mathcal{P}(\munderbar{P},\munderbar{b}) \oplus K\mathcal{P}(\bar{P},\bar{b}) \subseteq \mathcal{U}, \label{eq:input_inclusion_p} }
		\end{align}
	\end{subequations}
	it can be used to synthesize the RMPC scheme
	(since \eqref{eq:RPI_inclusion_p} implies \eqref{eq:RPI_property}, \eqref{eq:MPI_inclusion_p} implies \eqref{eq:PI_required}, \eqref{eq:state_inclusion_p}-\eqref{eq:input_inclusion_p} imply the constraint inclusions in \eqref{eq:smallRPI}-\eqref{eq:PI_required}).
	We encode \eqref{eq:RPI_inclusion_p}-\eqref{eq:input_inclusion_p} using Theorem~\ref{thm:inclusion_result}.
	%
	%
	%
	
	\begin{theorem}{\cite[Theorem 2]{Liu2019}}
		\label{thm:inclusion_result} For some $C \in \R^{n \times n^c}$, $M^c \in \R^{m^c \times n}$, $b^c \in \R^{m^c}_+$, $M^0 \in \R^{m^0 \times n}$, $b^o \in \R^{m^o}_+$, the
		inclusion $C \mathcal{P}(M^c,b^c) \subseteq \mathcal{P}(M^0,b^0)$
		holds if $\forall \ i \in \mathbb{I}_1^{m^0}, \exists L^c_{[i]} \in \mathbb{D}_+^{m^c}$ such that $\begin{bmatrix}
			2b_i^0 - b^{c^{\top}} L^c_{[i]} b^c &\hspace{-12pt} M^0_i C \\
			* &\hspace{-12pt} M^{c^{\top}} L^c_{[i]} M^c 
		\end{bmatrix} \succ 0.$ 
		$	\hfill\square $
	\end{theorem}
	\vspace{2pt}
	Hence,  $\eqref{eq:RPI_inclusion_p} \impliedby \forall i \in \mathbb{I}_1^{\munderbar{m}}, \exists  {\munderbar{D}}_{[i]} \in \D_+^{\munderbar{m}},{W}_{[i]} \in \D_+^{m_w}$ s.t.
	%
	\vspace{-6pt}
	\begin{equation}
		\label{eq:mRPI_NL}
		\scalemath{0.875}{
			\hspace{-15pt}
			\begin{matrix}
				&\begin{bmatrix}2{\munderbar{b}}_i - {\munderbar{b}}^{\top} {\munderbar{D}}_{[i]} {\munderbar{b}} - d^{\top}{W}_{[i]}d &\hspace{-5pt} {\munderbar{P}}_i &\hspace{-5pt} {\munderbar{P}}_i({A}+{B}{K}) \vspace{2pt}\\
					* &\hspace{-5pt} F^{\top}{W}_{[i]}F &\hspace{-5pt} \0 \\ * &\hspace{-5pt} * &\hspace{-5pt} {\munderbar{P}}^{\top} {\munderbar{D}}_{[i]} {\munderbar{P}} 
				\end{bmatrix} \succ 0,
		\end{matrix}}
	\end{equation}
	%
	$\eqref{eq:MPI_inclusion_p} \impliedby \forall \ i \in \mathbb{I}_1^{\bar{m}}, \exists \ {\bar{D}}_{[i]} \in \D_+^{\bar{m}}$ s.t.
	%
	\begin{equation}
		\label{eq:MPI_NL}
		\scalemath{0.87}{
			\hspace{-15pt}
			\begin{matrix}
				&\begin{bmatrix}2{\bar{b}}_i - {\bar{b}}^{\top} {\bar{D}}_{[i]} {\bar{b}} & \begin{matrix}{\bar{P}}_i({A}+{B}{K})\end{matrix} \\
					* & \begin{matrix} {\bar{P}}^{\top} {\bar{D}}_{[i]} {\bar{P}} \end{matrix}
				\end{bmatrix} \succ 0,
		\end{matrix}}
	\end{equation}
	$\eqref{eq:state_inclusion_p} \impliedby \forall \ i \in \mathbb{I}_1^{m_x}, \exists \ {\munderbar{S}}_{[i]} \in \D_+^{\munderbar{m}},{\bar{S}}_{[i]} \in \D_+^{\bar{m}}$ s.t.
	\begin{equation}
		\label{eq:state_feasible_NL}
		\scalemath{0.87}{
			\hspace{-15pt}
			\begin{matrix}
				&\begin{bmatrix}2v^x_i - {\munderbar{b}}^{\top} {\munderbar{S}}_{[i]} {\munderbar{b}} - {\bar{b}}^{\top} {\bar{S}}_{[i]} {\bar{b}} & V^x_i & V^x_i \vspace{2pt}\\
					* &  {\munderbar{P}}^{\top} {\munderbar{S}}_{[i]} {\munderbar{P}} & \0 \\ * & * & {\bar{P}}^{\top} {\bar{S}}_{[i]} {\bar{P}} 
				\end{bmatrix} \succ 0,
		\end{matrix}}
	\end{equation}
	$\eqref{eq:input_inclusion_p} \impliedby \forall \ i \in \mathbb{I}_1^{m_u}, \exists \ {\munderbar{R}}_{[i]} \in \D_+^{\munderbar{m}},{\bar{R}}_{[i]} \in \D_+^{\bar{m}}$ s.t.
	\begin{equation}
		\label{eq:input_feasible_NL}
		\scalemath{0.87}{
			\hspace{-15pt}
			\begin{matrix}
				&\begin{bmatrix}2v^u_i - {\munderbar{b}}^{\top} {\munderbar{R}}_{[i]} {\munderbar{b}} - {\bar{b}}^{\top} {\bar{R}}_{[i]} {\bar{b}} & V^u_i{K} & V^u_i{K} \vspace{2pt}\\
					* &  {\munderbar{P}}^{\top} {\munderbar{R}}_{[i]} {\munderbar{P}} & \0 \\ * & * & {\bar{P}}^{\top} {\bar{R}}_{[i]} {\bar{P}} 
				\end{bmatrix} \succ 0.
		\end{matrix}}
	\end{equation}
	We now formulate a criterion to select the variables formulating \eqref{eq:mRPI_NL}-\eqref{eq:input_feasible_NL} along with $(A,B,d) \in \hat{\Sigma}_T$, leading to an optimization problem.
	In this formulation, we assume that the matrices $(\munderbar{P},\bar{P},F)$ are known a priori. While this assumption increases conservativeness in our approach, it simplifies the solution procedure. We note that a good set of hyperplanes $(\munderbar{P},\bar{P})$ can \textit{guessed} for some initial $(A,B,F,d)$ using \cite{Liu2019}, and kept constant for our approach. Moreover, our approach can be extended to optimize over $(\munderbar{P},\bar{P})$ using the results in \cite{LiuThesis}. We skip further details here due to space limitations.
		\begin{remark}
			The condition in 
			Theorem \ref{thm:inclusion_result} is necessary and sufficient for the inclusion $C \mathcal{P}(M^c,b^c) \subseteq \mathcal{P}(M^0,b^0)$ if a non-strict inequality $\succeq$ is used. However, we only use the sufficiency property given by $\succ$ for numerical robustness. $	\hfill\square $
	\end{remark}
	%
	%
	%
	%
	
	3) \textbf{\textit{Identification criterion:}}
	For RMPC synthesis, it is desirable to compute a small RPI set $\Delta \mathcal{X}$ to reduce constraint tightening,
	and to regulate the system to a small neighborhood of the origin \cite{Mayne2005}. Hence, we minimize $\norm{\munderbar{b}}_1$, since it corresponds to the smallest (in an inclusion sense) RPI set represented by fixed hyperplanes $\munderbar{P}$ \cite[Corollary 1]{Rakovic2013}.
%
	%
	%
	Moreover, we know from \cite[Proposition 2]{Mayne2005} that a large terminal set $\mathcal{X}_{\mathrm{t}}$ maximizes the region of attraction $\Omega_N$. Hence, we maximize the size of $\mathcal{X}_{\mathrm{t}}$ by minimizing a distance metric between $\mathcal{X}_{\mathrm{t}}$ and the state constraint set $\mathcal{X}$ as follows: let  $\mathcal{B}(\bar{\epsilon}):=\mathcal{P}(\bar{E},\bar{\epsilon}) \subset \R^{n_x}$ with $\bar{\epsilon} \in \R^{m_{\epsilon}}_+$ and $\bar{E}$ fixed a priori; then, 
	we minimize $\norm{\bar{\epsilon}}_1$ subject to the inclusion $\mathcal{X}\subseteq \mathcal{P}(\bar{P},\bar{b})\oplus\mathcal{B}(\bar{\epsilon})$. Assuming to know the vertices $\{x_{[i]}, i \in \mathbb{I}_1^{m_x^v}\}$ of $\mathcal{X}$,  $(\bar{b},\bar{\epsilon}) \in \bar{\mathcal{S}}$ implies $\mathcal{X}\subseteq \mathcal{P}(\bar{P},\bar{b})\oplus\mathcal{B}(\bar{\epsilon})$ where $\bar{\mathcal{S}}:=\{
		(\bar{b},\bar{\epsilon}): \ 
		 x_{[i]} \in \mathcal{P}(\bar{P},\bar{b}) \oplus \mathcal{P}(\bar{E},\bar{\epsilon}), \forall \ i \in \mathbb{I}_1^{m_x^v}\}.$
	
	Finally, since the performance matrices $(H_{\mathrm{Q}},P_{\mathrm{Q}})$ formulating the RMPC controller in \eqref{eq:RMPC_controller} are fixed by $(A,B,K)$ as noted in Remark \ref{rem:FB_performance}, we introduce a way to tune the closed-loop performance: We 
	evaluate the performance
	using the system $\hat{x}(t+1)=A\hat{x}(t)+B\hat{u}(t)$ inside the terminal set as
	%
	\begin{align}
		\label{eq:performance_bound}
		\scalemath{0.95}{
			\hspace{-8pt}
			\hat{x}(0) \in \mathcal{X}_{\mathrm{t}}, \ \hat{u}(t)=K\hat{x}(t), \
			\sum_{t=0}^{\infty} \norm{\hat{x}(t)}^2_{\tilde{Q}}+\norm{\hat{u}(t)}^2_{\tilde{R}} \leq {\tilde{r}},}
	\end{align}
	where $\tilde{Q} \in \So_+^{n_x}$ and $\tilde{R}\in\So_+^{n_u}$ are user-defined performance matrices,	
	and we minimize $\tilde{r}$. Then, if ${\tilde{\Theta}}\in \So_+^{n_x}$ satisfies 
	\vspace{-1pt}
	\begin{align}
		\label{eq:dissipativity_NL}
		\scalemath{0.92}{
			({A}+{BK})^{\top}{\tilde{\Theta}}({A}+{BK})-{\tilde{\Theta}}+\tilde{Q}+{K}^{\top} \tilde{R} {K} \prec 0,}	
	\end{align}
	the left-hand-side of the inequality in \eqref{eq:performance_bound} is upper bounded by $\norm{\hat{x}(0)}^2_{{\tilde{\Theta}}}$ \cite{Kothare1996}. Hence, \eqref{eq:performance_bound} is satisfied if the inclusion $\mathcal{P}(\bar{P},{\bar{b}}) \subseteq \mathcal{E}({\tilde{\Theta}},\tilde{{r}})$ holds, thus imposing an upper bound on the size of the terminal set.
	Following the S-procedure \cite[Section 2.6.3]{lmi_book},
	%
	the inclusion $\mathcal{P}(\bar{P},{\bar{b}}) \subseteq \mathcal{E}({\tilde{\Theta}},\tilde{{r}})$ holds if
	\begin{equation}
		\label{eq:performance_ellipse_NL}
		\scalemath{0.92}{
			\begin{matrix}
				\exists {\tilde{M}} \in \D_+^{\bar{m}} \quad \text{s.t.} \quad {\bar{P}}^{\top} {\tilde{M}} {\bar{P}} - {\tilde{\Theta}} \succ 0, \quad {\tilde{r}} - {\bar{b}}^{\top} {\tilde{M}} {\bar{b}}>0.
		\end{matrix}}
	\end{equation}
	Based on these considerations, we formulate the identification problem as the following NLPMI
	%
	\begin{align}
		\label{eq:main_SDP}
		&\min_{{Z}_{\mathrm{NL}}} \
		\alpha{\norm{\munderbar{b}}_1}+
		\beta{\norm{\bar{\epsilon}}_1}+\gamma \tilde{r} \\
		& \quad \text{s.t.} \ ({A},{B},{d}) \in \hat{\Sigma}_T, \ (\bar{b},\bar{\epsilon}) \in \bar{\mathcal{S}}, \ \eqref{eq:mRPI_NL}-\eqref{eq:input_feasible_NL}, \eqref{eq:dissipativity_NL}-\eqref{eq:performance_ellipse_NL}\nonumber 
	\end{align}
	where $\alpha,\beta,\gamma\geq0$ are user-defined weights, and
	\begin{align*}
		\scalemath{0.9}{
			{Z}_{\mathrm{NL}}:=\begin{pmatrix}
				A,B,d,\mathcal{Z},\lambda,K,\munderbar{b},\bar{b},\tilde{\Theta},\tilde{r},\tilde{M},\bar{\epsilon},
				\\ \{\munderbar{D}_{[i]},{{W}}_{[i]},i \in \mathbb{I}_1^{\munderbar{m}}\},
				\{{\bar{D}}_{[i]}, i \in \mathbb{I}_1^{\bar{m}}\},\\ \{\munderbar{S}_{[i]},\bar{S}_{[i]},i\in \mathbb{I}_1^{m_x}\},
				\{\munderbar{R}_{[i]},\bar{R}_{[i]},i \in \mathbb{I}_1^{m_u}\}
		\end{pmatrix}}.
	\end{align*}

	4) \textbf{\textit{Feasible SCP for Problem \eqref{eq:main_SDP}}}:
	In order to solve Problem \eqref{eq:main_SDP}, a standard 
	SCP approach can be adopted, in which a sequence of SDPs approximating \eqref{eq:main_SDP} are solved. However, to guarantee feasibility of the iterates, we adopt the following SCP procedure.
	%
	Starting from an initial feasible iterate $Z_{\mathrm{NL}}$,
	we solve a sequence of SDPs formulated using sufficient LMI conditions for the constraints of Problem \eqref{eq:main_SDP}, such that the method produces feasible iterates. The sufficient LMI conditions
		are formulated using convex underestimates \cite{TranDinhQuoc2012SCPa} of the NLMI constraints at the current iterate. Moreover, the objective value of \eqref{eq:main_SDP} is non-increasing over the iterates, such that globalization is unnecessary and we terminate when the objective value does not reduce further.

	\noindent \vspace{5pt}\\
	{\textbf{($a$) Convex SDP approximation}}:
	Given a feasible iterate $Z_{\mathrm{NL}}$ for Problem \eqref{eq:main_SDP}, we formulate sufficient LMI conditions for \eqref{eq:mRPI_NL}-\eqref{eq:input_feasible_NL},~\eqref{eq:dissipativity_NL},~\eqref{eq:performance_ellipse_NL} using the following result.
	\begin{proposition}{\cite[Lemma 2.1]{Liu2019}}
		\label{prop:linearization}
		Let matrices $\bm{L},L \in \R^{m \times n}$ and $\bm{D},D \in \So^m_+$, and define the matrix functions $	\mathcal{L}^{L,D}_{\bm{L},\bm{D}}:=\bm{L}^{\top}D^{-1}L+L^{\top}D^{-1}\bm{L}-L^{\top}D^{-1}\bm{D}D^{-1}L,$ and $\mathcal{N}_{\bm{L},\bm{D}}:=\bm{L}^{\top}\bm{D}^{-1}\bm{L}$.
		Then, $\mathcal{N}_{\mathbf{L},\mathbf{D}}\succeq \mathcal{L}^{L,D}_{\mathbf{L},\mathbf{D}}$ and $\mathcal{N}_{{L},{D}}=\mathcal{L}^{L,D}_{{L},{D}}$. Hence, if $ \exists(L,D)$ such that $\mathcal{N}_{{L},{D}}\succ 0$, then
		$\exists(\bm{L},\bm{D})$ such that $\mathcal{N}_{\bm{L},\bm{D}} \succeq \mathcal{L}^{L,D}_{\bm{L},\bm{D}} \succ 0$. $\hfill\square$
	\end{proposition}
	This result implies that if $\mathcal{N}_{L,D} \succ 0$, then the LMI $\mathcal{L}^{L,D}_{\bm{L},\bm{D}} \succ 0$ is a convex underestimate and a sufficient condition for $\mathcal{N}_{\bm{L},\bm{D}} \succ 0$.
	We will now use this property to formulate sufficient LMIs for \eqref{eq:mRPI_NL}-\eqref{eq:input_feasible_NL},~\eqref{eq:dissipativity_NL},~\eqref{eq:performance_ellipse_NL}.
	%
	%
	The claimed SCP feasibility and cost decrease are then obtained as a corollary.
	%
	%
	%
	%
	
	\begin{theorem}
		\label{thm:main_result}
		Suppose that $Z_{\mathrm{NL}}$ is feasible for \eqref{eq:main_SDP}. Then:
		
		($i$) \textit{RPI condition \eqref{eq:mRPI_NL}}: For each $ i \in \mathbb{I}_1^{\munderbar{m}}$, there exists
		$(\bm{A},\bm{B},\bm{d},\bm{K},\munderbar{\bm{b}}$,$\hat{\munderbar{\bm{D}}}_{[i]},\hat{{\bm{W}}}_{[i]})$ 
		satisfying the LMI
		\begin{align}
			\label{eq:mRPI_LMI}
			&\scalemath{0.83}{
				\hspace{10pt}
				\begin{bmatrix}
					\I &\hspace{-10pt} \0 &\hspace{-10pt} \0 &\hspace{-10pt} -\bm{B}^{\top}\munderbar{P}_i^{\top} &\hspace{-10pt} \0 &\hspace{-10pt} \bm{K} \\
					* &\hspace{-5pt} \hat{\munderbar{\bm{D}}}_{[i]} &\hspace{-10pt} \0 &\hspace{-10pt} \munderbar{\bm{b}} &\hspace{-10pt} \0 &\hspace{-10pt} \0 \\
					* &\hspace{-10pt} * &\hspace{-5pt}  \hat{{\bm{W}}}_{[i]} &\hspace{-10pt} \bm{d} &\hspace{-10pt} \0 &\hspace{-10pt} \0  \\
					* &\hspace{-10pt} * &\hspace{-10pt} * &\hspace{-10pt} 2\munderbar{\bm{b}}_i+ \mathcal{L}_{\bm{B}^{\top}\munderbar{P}_i^{\top},\I}^{{B}^{\top}\munderbar{P}_i^{\top},\I} &\hspace{-10pt} \munderbar{P}_i &\hspace{-10pt} \munderbar{P}_i\bm{A} \\
					* &\hspace{-10pt} * &\hspace{-10pt} * &\hspace{-10pt} * &\hspace{-10pt} \mathcal{L}_{F,\hat{{\bm{W}}}_{[i]}}^{F,{W_{[i]}^{-1}}} &\hspace{-10pt} \0 \\
					* &\hspace{-10pt} * &\hspace{-10pt} * &\hspace{-10pt} * &\hspace{-10pt} * &\hspace{-10pt} \mathcal{L}_{\munderbar{P}\hat{\munderbar{\bm{D}}}_{[i]}}^{\munderbar{P},{\munderbar{D}_{[i]}^{-1}}}+\mathcal{L}_{\bm{K},\I}^{K,\I}
				\end{bmatrix} \succ 0},
		\end{align}
		and 
		$(\bm{A},\bm{B},\bm{d},\bm{K},\munderbar{\bm{b}},\hat{\munderbar{\bm{D}}}^{-1}_{[i]},\hat{{\bm{W}}}^{-1}_{[i]})$ satisfies \eqref{eq:mRPI_NL}.

		($ii$) \textit{PI condition \eqref{eq:MPI_NL}}: 
		For each $i \in \mathbb{I}_1^{\bar{m}}$, there exists 
		$(\bm{A},\bm{B},\bm{K},\bar{\bm{b}},\hat{\bar{\bm{D}}}_{[i]})$ 
		satisfying the LMI
		%
		\begin{align}
			\label{eq:MPI_LMI}
			&\scalemath{0.83}{
				\begin{bmatrix}
					\I & \0 & -\bm{B}^{\top}\bar{P}_i^{\top}  & \bm{K} \\
					* & \hat{\bar{\bm{D}}}_{[i]}  & \bar{\bm{b}} & \0 \\
					* & * & 2\bar{\bm{b}}_i+ \mathcal{L}_{\bm{B}^{\top}\bar{P}_i^{\top},\I}^{{B}^{\top}\bar{P}_i^{\top},\I} & \bar{P}_i\bm{A} \\
					* & * & * & \mathcal{L}_{\bar{P},\hat{\bar{\bm{D}}}_{[i]}}^{\bar{P},{\bar{D}_{[i]}^{-1}}}+\mathcal{L}_{\bm{K},\I}^{K,\I}
				\end{bmatrix} \succ 0},
		\end{align}
		and 
		$(\bm{A},\bm{B},\bm{K},\bar{\bm{b}},\hat{\bar{\bm{D}}}^{-1}_{[i]})$ satisfies \eqref{eq:MPI_NL}. 
		
		($iii$) \textit{Constraint inclusions \eqref{eq:state_feasible_NL},~\eqref{eq:input_feasible_NL}}:
		For each $i \in \mathbb{I}_1^{m_x}$, there exists 
		$(\munderbar{\bm{b}},\bar{\bm{b}},\hat{\munderbar{\bm{S}}}_{[i]},\hat{\bar{\bm{S}}}_{[i]})$,
		and for each  $i \in \mathbb{I}_1^{m_u}$, there exists
		$(\bm{K},\munderbar{\bm{b}},\bar{\bm{b}},\hat{\munderbar{\bm{R}}}_{[i]},\hat{\bar{\bm{R}}}_{[i]})$ 
		satisfying the LMIs
		\vspace{-10pt}\\
		%
		\begin{align}
			\label{eq:state_feasible_LMI}
			\scalemath{0.83}{
				\begin{bmatrix}
					\hat{\munderbar{\bm{S}}}_{[i]} & \0 & \munderbar{\bm{b}} & \0 & \0\\
					* & \hat{\bar{\bm{S}}}_{[i]} & \bar{\bm{b}} & \0 & \0\\
					* & * & 2v^x_i  & V^x_i & V^x_i \\
					* & * & * & \mathcal{L}_{{\munderbar{P}},\hat{\munderbar{\bm{S}}}_{[i]}}^{{\munderbar{P}},{\munderbar{S}}_{[i]}^{-1}} & \0 \\
					* & * & * & * & \mathcal{L}_{{\bar{P}},\hat{\bar{\bm{S}}}_{[i]}}^{{\bar{P}},{\bar{S}}_{[i]}^{-1}}
				\end{bmatrix} \succ 0}, 
		\end{align}
		\vspace{-25pt}\\
		\begin{align}
			\label{eq:input_feasible_LMI}
			\scalemath{0.83}{
				\begin{bmatrix}
					\hat{\munderbar{\bm{R}}}_{[i]} & \0 & \munderbar{\bm{b}} & \0 & \0\\
					* & \hat{\bar{\bm{R}}}_{[i]} & \bar{\bm{b}} & \0 & \0\\
					* & * & 2v^u_i  & V^u_i\bm{K} & V^u_i\bm{K} \\
					* & * & * & \mathcal{L}_{{\munderbar{P}},\hat{\munderbar{\bm{R}}}_{[i]}}^{{\munderbar{P}},{\munderbar{R}}_{[i]}^{-1}} & \0 \\
					* & * & * & * & \mathcal{L}_{{\bar{P}},\hat{\bar{\bm{R}}}_{[i]}}^{{\bar{P}},{\bar{R}}_{[i]}^{-1}}
				\end{bmatrix} \succ 0},
		\end{align}
		and
		$\scalemath{0.98}{(\munderbar{\bm{b}},\bar{\bm{b}},\hat{\munderbar{\bm{S}}}^{-1}_{[i]},\hat{\bar{\bm{S}}}^{-1}_{[i]})}$ 
		, $\scalemath{0.98}{(\bm{K},\munderbar{\bm{b}},\bar{\bm{b}},\hat{\munderbar{\bm{R}}}^{-1}_{[i]},\hat{\bar{\bm{R}}}^{-1}_{[i]})}$ satisfy \eqref{eq:state_feasible_NL}, \eqref{eq:input_feasible_NL}.
		
		%
		($iv$) \textit{Dissipativity condition \eqref{eq:dissipativity_NL}}: 
		There exists $(\bm{A},\bm{B},\bm{K},\tilde{\bm{\Theta}})$ satisfying the LMI
		\vspace{-2pt}
		\begin{align}
			\label{eq:dissipativity_LMI}
			&\scalemath{0.83}{
				\begin{bmatrix}
					\I & \0 & \0 & -\bm{B}^{\top} & \bm{K} \\
					* & \tilde{Q}^{-1} & \0 & \0 & \I  \\
					* & * &  \tilde{R}^{-1} &  \0 & \bm{K}  \\
					* & * & * & \mathcal{L}_{\I,\tilde{\bm{\Theta}}}^{\I,\tilde{{\Theta}}}+\mathcal{L}_{\bm{B}^{\top},\I}^{{B}^{\top},\I} & \bm{A} \\
					* & * & * & * & \tilde{\bm{\Theta}}+\mathcal{L}_{\bm{K},\I}^{K,\I}
				\end{bmatrix} \succ 0},
		\end{align}
		and 
		$(\bm{A},\bm{B},\bm{K},\tilde{\bm{\Theta}})$ satisfies \eqref{eq:dissipativity_NL}.
		
		($v$) \textit{Performance ellipsoid inclusion condition \eqref{eq:performance_ellipse_NL}}: 
		There exists $(\bar{\bm{b}}, \tilde{\bm{\Theta}}, \tilde{\bm{r}},\hat{\tilde{\bm{M}}})$ satisfying the LMI
		\begin{align}
			\hspace{45pt}
			\label{eq:performance_ellipse_LMI}
			&\mathcal{L}_{\bar{P},\hat{\tilde{\bm{M}}}}^{\bar{P},\tilde{{M}}^{{-1}}} - \tilde{\bm{\Theta}} \succ 0, \quad 
			\begin{bmatrix} \hat{\tilde{\bm{M}}} & \bar{\bm{b}} \\ * & \tilde{\bm{r}} \end{bmatrix} \succ 0, 
		\end{align}
		and 
		$(\bar{\bm{b}},\tilde{\bm{\Theta}},\tilde{\bm{r}},\hat{\tilde{\bm{M}}}^{-1})$ satisfies \eqref{eq:performance_ellipse_NL}.
		%
	\end{theorem}
	\begin{proof}
		The proof follows by using the Schur complement and Proposition \ref{prop:linearization} on \eqref{eq:mRPI_NL}-\eqref{eq:input_feasible_NL}. We detail the proof of Part $(i)$, since Parts $(ii)$-$(v)$ follow with similar arguments.\\
		Part ($i$) : As $(A,B,K,\munderbar{b},d,\munderbar{D}_{[i]},W_{[i]})$ in $Z_{\mathrm{NL}}$ satisfy \eqref{eq:mRPI_NL}, we take a Schur complement of the $(1,1)$ block to obtain
		\begin{align}
			\label{eq:mRPI_step1}
			\scalemath{0.88}{
				\begin{bmatrix}
					\munderbar{{D}}^{{-1}}_{[i]} & \0 & \munderbar{{b}} & \0 & \0 \\
					* &  W^{-1}_{[i]} & {d} & \0 & \0  \\
					* & * & 2\munderbar{b}_i & \munderbar{P}_i & \munderbar{P}_i{A}+\munderbar{P}_i BK \\
					* & * & * & F^{\top}W_{[i]}F & \0 \\
					* & * & * & * & \munderbar{P}^{\top} \munderbar{D}_{[i]} \munderbar{P}
				\end{bmatrix} \succ 0.
			}
		\end{align}
			Defining $\hat{W}_{[i]}:=W_{[i]}^{-1}$ and $\hat{\munderbar{D}}_{[i]}:=\munderbar{D}_{[i]}^{-1}$, 
			Eq.~\eqref{eq:mRPI_step1} is nonlinear in variables $(B,K,\hat{\munderbar{D}}_{[i]},\hat{W}_{[i]})$ in 
			the blocks $(4,4)$, $(5,5)$, $(3,5)$ and $(5,3)$. 
		%
		Then, we write \eqref{eq:mRPI_step1} as
		\begin{align}
			\label{eq:mRPI_step2}
			\scalemath{0.88}{
				\begin{bmatrix}
					\hat{\munderbar{{D}}}_{[i]} & \0 & \munderbar{{b}} & \0 & \0 \\
					* &  \hat{W}_{[i]} & d & \0 & \0  \\
					* & * & \munderbar{\mathcal{N}}_{i,33} & \munderbar{P}_i & \munderbar{P}_i{A} \\
					* & * & * & \mathcal{N}_{F,\hat{W}_{[i]}} & \0 \\
					* & * & * & * & \munderbar{\mathcal{N}}_{i,55}
				\end{bmatrix}-\munderbar{\mathcal{K}}_i^{\top} \munderbar{\mathcal{K}}_i \succ 0, 
			}
		\end{align}
		where $\munderbar{\mathcal{N}}_{i,33}:=2\munderbar{b}_i+\mathcal{N}_{B^{\top}\munderbar{P}_i^{\top},\I}$, $\munderbar{\mathcal{N}}_{i,55}:=\mathcal{N}_{\munderbar{P},\hat{\munderbar{D}}_{[i]}}+
		\mathcal{N}_{K,\I}$, and $\munderbar{\mathcal{K}}_i:=[\0 \ \ \0 \ \ -B^{\top} \munderbar{P}_i^{\top} \ \ \0 \ \ K]$ (with the function $\mathcal{N}_{.,.}$ defined in Proposition \ref{prop:linearization}). Taking Schur complement of \eqref{eq:mRPI_step2},
		%
		\begin{align}
			\label{eq:mRPI_step3}
			\scalemath{0.88}{
				\hspace{-5pt}
				\begin{bmatrix}
					\I &\hspace{-5pt} \0 &\hspace{-5pt} \0 &\hspace{-5pt} -{B}^{\top}\munderbar{P}_i^{\top} &\hspace{-5pt} \0 &\hspace{-5pt} {K} \\
					* &\hspace{-5pt} \hat{\munderbar{{D}}}_{[i]} &\hspace{-5pt} \0 &\hspace{-5pt} \munderbar{{b}} &\hspace{-5pt} \0 &\hspace{-5pt} \0 \\
					* &\hspace{-5pt} * &\hspace{-5pt}  \hat{W}_{[i]} &\hspace{-5pt} d &\hspace{-5pt} \0 &\hspace{-5pt} \0  \\
					* &\hspace{-5pt} * &\hspace{-5pt} * &\hspace{-5pt} 2\munderbar{b}_i+\mathcal{N}_{B^{\top}\munderbar{P}_i^{\top},\I} &\hspace{-5pt} \munderbar{P}_i &\hspace{-5pt} \munderbar{P}_i{A} \\
					* &\hspace{-5pt} * &\hspace{-5pt} * &\hspace{-5pt} * &\hspace{-5pt} \mathcal{N}_{F,\hat{W}_{[i]}} &\hspace{-5pt} \0 \\
					* &\hspace{-5pt} * &\hspace{-5pt} * &\hspace{-5pt} * & \hspace{-5pt} * &\hspace{-5pt} \mathcal{N}_{\munderbar{P},\hat{\munderbar{D}}_{[i]}}+
					\mathcal{N}_{K,\I}
				\end{bmatrix} \succ 0}
		\end{align}
		results, with all nonlinear components collected in the diagonal blocks.
		Using Proposition~\ref{prop:linearization} on these components, we  conclude that \eqref{eq:mRPI_LMI} is a sufficient LMI condition for \eqref{eq:mRPI_step3}.
	\end{proof}
	\begin{corollary}
		\label{corr:update_corollary}
		Suppose that $Z_{\mathrm{NL}}$ is feasible for Problem \eqref{eq:main_SDP}. 
		Then, the solution of the SDP
		\begin{align}
			\label{eq:update_SDP}
			\scalemath{0.98}{
				\begin{matrix}
					\hspace{-70pt} \min_{\bm{Z}} \ \alpha \norm{\munderbar{\bm{b}}}_1+\beta \norm{\bar{\bm{\epsilon}}}_1+\gamma \tilde{\bm{r}} \vspace{3pt} \vspace{2pt}\\
					\qquad \text{ s.t. } (\bm{A},\bm{B},\bm{d}) \in \hat{\Sigma}_T, 
					\ \ (\bar{\bm{b}},\bar{\bm{\epsilon}}) \in \bar{\mathcal{S}}, 
					\ \ \eqref{eq:mRPI_LMI}-\eqref{eq:performance_ellipse_LMI}, \vspace{-7pt}\\
			\end{matrix}}
		\end{align}
		\begin{align*}
			\vspace{-5pt}
			\scalemath{0.9}{
				{\bm{Z}}:=\begin{pmatrix}
					\bm{A},\bm{B},\bm{d},\bm{\mathcal{Z}},\bm{\lambda},\bm{K},\munderbar{\bm{b}},\bar{\bm{b}},\tilde{\bm{\Theta}}, \tilde{\bm{r}},\hat{\tilde{\bm{M}}},\bar{\bm{\epsilon}},\\
					\{\hat{\munderbar{\bm{D}}}_{[i]},{{\hat{\bm{W}}}}_{[i]},i \in \mathbb{I}_1^{\munderbar{m}}\},
					\{{\hat{\bar{\bm{D}}}}_{[i]}, i \in \mathbb{I}_1^{\bar{m}}\},\\
					\{\hat{\munderbar{\bm{S}}}_{[i]},\hat{\bar{\bm{S}}}_{[i]},i\in \mathbb{I}_1^{m_x}\},
					\{\hat{\munderbar{\bm{R}}}_{[i]},\hat{\bar{\bm{R}}}_{[i]},i \in \mathbb{I}_1^{m_u}\}
			\end{pmatrix}},
		\end{align*} 
		is feasible for Problem  \eqref{eq:main_SDP}, and satisfies the cost decrease condition $\alpha \norm{\munderbar{\bm{b}}}_1+\beta \norm{\bar{\bm{\epsilon}}}_1+\gamma \tilde{\bm{r}}\leq \alpha\norm{\munderbar{{b}}}_1+\beta \norm{\bar{{\epsilon}}}_1+\gamma \tilde{r}.$
		$\hfill\square$
	\end{corollary}
	\begin{proof}
		The feasibility of $\bm{Z}$ for Problem \eqref{eq:main_SDP} follows from Theorem \ref{thm:main_result}, and the cost decrease condition holds since $Z_{\mathrm{NL}}$ is feasible for Problem \eqref{eq:update_SDP}.
	\end{proof}
	\noindent \vspace{-8pt}\\ 
	We propose the following procedure to solve Problem \eqref{eq:main_SDP}.
	\begin{align*}
		\scalemath{0.98}{
			\hspace{-10pt}
			\begin{matrix*}[l]
				&\text{\underbar{Algorithm $1$ : Update solution of Problem \eqref{eq:main_SDP}}} \\
				&\text{1. Obtain an initial feasible solution $Z_{\mathrm{NL}}$ for Problem \eqref{eq:main_SDP}.} \\
				&\text{2. Solve the SDP in \eqref{eq:update_SDP} for the updated variables $\bm{Z}$;} \\
				&\hspace{12pt}\text{Recover feasible values $Z_{\mathrm{NL}}$ from the solution.} \\
				&\text{3. Evaluate the objective value $\alpha\norm{\munderbar{{b}}}_1+\beta \norm{\bar{{\epsilon}}}_1+\gamma \tilde{r}$;} 
				\vspace{-0pt}\\
				&\hspace{0pt}\text{4. If there is a reduction from previous iteration, repeat}\\
				&\hspace{12pt}\text{Step 2 using $Z_{\mathrm{NL}}$ for linearization. Else, terminate. } 
				\hfill\square
		\end{matrix*} }
	\end{align*} 
		We note that the results in \cite[Chapter 4]{TranDinhQuoc2012SCPa} can be used to study the convergence of Algorithm $1$. Further analysis is beyond the scope of the current work.
	
	{\textbf{($b$) Initialization procedure}}: 
	We propose the following procedure 
	to compute an initial feasible solution $Z_{\mathrm{NL}}$.
	\\
	($i$) Select some $\hat{\theta}_T>0$ through a guess to characterize $\hat{\Sigma}_T$.
		%
		\\	
		($ii$) Solve the LP $\arg\min_{{A},{B},{d}} \{ \norm{d}_1 \text{ s.t. } ({A},{B},{d}) \in \hat{\Sigma}_T\}$ for an initial feasible model $(A,B,d)$.
		\\
		($iii$) Use the method in \cite{Liu2019} to compute an initial RPI set $\Delta \mathcal{X}=\mathcal{P}(\munderbar{P},\munderbar{b})$ satisfying \eqref{eq:RPI_inclusion_p} along with a feedback gain $K$, while enforcing $\mathcal{P}(\munderbar{P},\munderbar{b})\subset \mathcal{X}$ and $K\mathcal{P}(\munderbar{P},\munderbar{b})\subset \mathcal{U}$. \\
		($iv$) Compute the tightened constraint set $\mathcal{O}_0:=\{x: x \in \mathcal{X} \ominus \Delta \mathcal{X}, Kx \in \mathcal{U} \ominus K \Delta \mathcal{X}\}$, and then compute a PI set $\mathcal{X}_{\mathrm{t}}=\mathcal{P}(\bar{P},\bar{b})$ using the method in \cite{Liu2019} for $x(t+1)=(A+BK)x(t)$.\\
		($v$) Compute the remaining variables formulating Problem \eqref{eq:main_SDP} by solving $\min_{{Z}_{\mathrm{I}}} \{\tilde{r}:$  \eqref{eq:mRPI_NL}--\eqref{eq:input_feasible_NL}, \eqref{eq:dissipativity_NL}--\eqref{eq:performance_ellipse_NL}$\}$, where
		\begin{align*}
			\scalemath{0.88}{
				{Z}_{\mathrm{I}}:=\begin{pmatrix}{\{\munderbar{D}}_{[i]},{{W}}_{[i]},i \in \mathbb{I}_1^{\munderbar{m}}\},\{{\bar{D}}_{[i]}, i \in \mathbb{I}_1^{\bar{m}}\},\{\munderbar{S}_{[i]},\bar{S}_{[i]},i\in \mathbb{I}_1^{m_x}\}, \\
					\{\munderbar{R}_{[i]},\bar{R}_{[i]},i \in \mathbb{I}_1^{m_u}\},\tilde{\Theta},\tilde{r},\tilde{M}
			\end{pmatrix}}.
		\end{align*}
		\begin{remark}
			In Steps ($ii$),($iii$), the methods in \cite{Liu2019} guarantee the feasibility of the SDP in Step ($v$), since they are also formulated using Theorem \ref{thm:inclusion_result}. Other methods, e.g. \cite{Trodden2016,Gilbert1991}, can also be used if the feasibility of Step ($v$) is ensured.
			%
			$\hfill\square$
		\end{remark}
	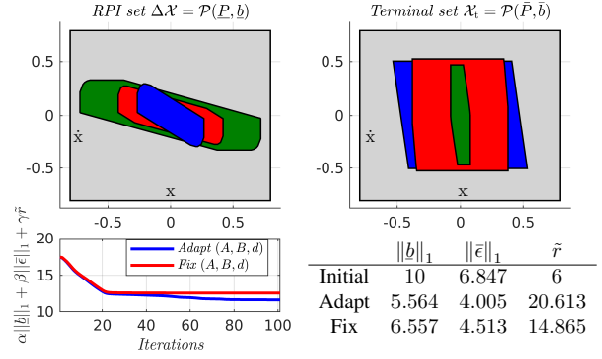
\begin{figure}[t]
		\centering
		\vspace{-0.2cm}
		\hspace{-1cm}
		\resizebox{.48\textwidth}{!}{
			\begin{tikzpicture}
				\begin{scope}[shift={(-10,0)}]
					\node[draw=none,fill=none](sets_fig) at (0,0) {\includegraphics[trim=78 301 110 301,clip,scale=0.4]{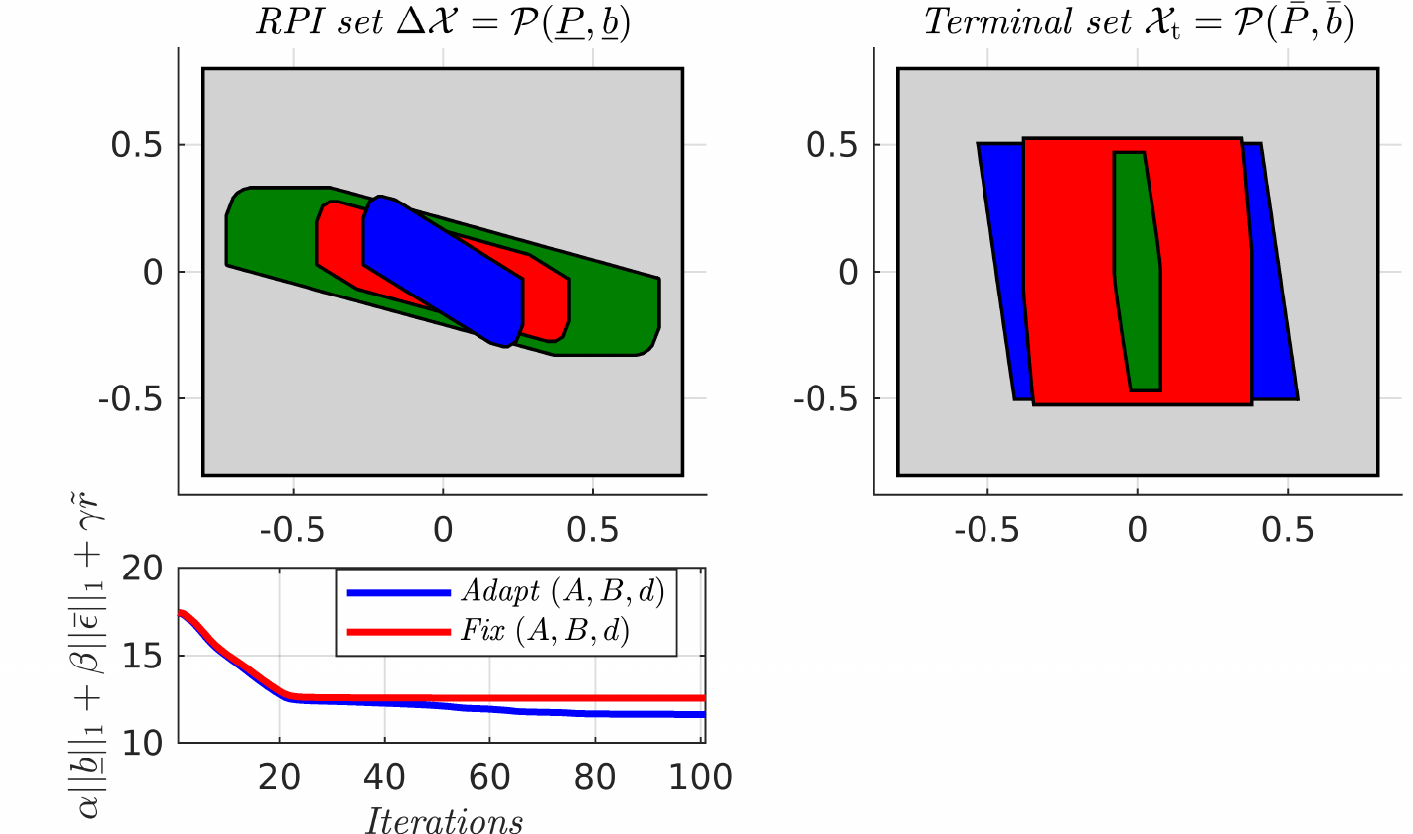}};
					\node[label=above:](Y5_set_left) [above left=-0.25cm and -2.2cm of sets_fig]{};
					\node[label=above:](Y5_set_right) [right=0.15cm of Y5_set_left]{};
					\node[](mRPI_label) [right=-0.35cm of Y5_set_right] {};
					\node[](MPI_label) [right=3.4cm of mRPI_label] {};
					\node[](xdot_left) [below left=1cm and 1.05cm of mRPI_label] {$\scalemath{0.5}{\dot{\mathrm{x}}}$};
					\node[](x_left) [below right=0.2cm and 0.56cm of xdot_left] {$\scalemath{0.5}{{\mathrm{x}}}$};
					\node[](xdot_left) [below left=1cm and 1.85cm of MPI_label] {$\scalemath{0.5}{\dot{\mathrm{x}}}$};
					\node[](x_left) [below right=0.2cm and 0.56cm of xdot_left] {$\scalemath{0.5}{{\mathrm{x}}}$};
					\node[](tablee) [below right=0.7cm and -0.8cm of xdot_left] {${\scalemath{0.58}{\begin{matrix}&\norm{\munderbar{b}}_1 & \norm{\bar{\epsilon}}_1 & \tilde{r} \\ \hline
									\text{Initial} & 10 & 6.847 & 6 \\
									\text{Adapt}  & 5.564 & 4.005 & 20.613  \\
									\text{Fix} & 6.557 & 4.513 & 14.865\end{matrix}}}$};
				\end{scope}
		\end{tikzpicture}}
		\captionsetup{width=1\linewidth}
		\caption{Results of Algorithm~1. Gray sets-$\mathcal{X}$, Green sets-Initialization, Blue sets-$(A,B,d)$ as optimization variables, Red sets - Fix $(A,B,d)$ to initial values.
			%
			%
			%
			%
		} \label{fig:main_result}
	\end{figure}
	\section{Numerical example}
	We consider a nonlinear mass-spring-damper system with dynamics $\mathrm{F}=\mathrm{m}\ddot{\mathrm{x}}+(\mathrm{K}\mathrm{x}+\mathrm{K}_{\mathrm{NL}}\mathrm{x}^2)+(\mathrm{c}\dot{\mathrm{x}}+\mathrm{c}_{\mathrm{NL}}\dot{\mathrm{x}}^2)+\mathrm{F}_{\delta}$, where
	%
	$u = \mathrm{F}$, $x = [\mathrm{x} \ \dot{\mathrm{x}}]^{\top}$, and constraints $\mathcal{X}=\{x:||x||_{\infty}\leq 0.8\}$, $\mathcal{U}=\{u:||u||_{\infty}\leq 2.5\}$.
	We simulate the plant using $\mathrm{ode45}$ integration to build the dataset $\mathcal{D}$ with $T=1000$ at a $0.1s$ time interval. We set $\mathrm{K}_{\mathrm{NL}},\mathrm{c}_{\mathrm{NL}}=0.12$, and uniformly sample the parameters $\mathrm{m},\mathrm{K},\mathrm{c}$ in $(0.44,0.56)$ and $\mathrm{F}_{\delta}$ in $(-0.12,0.12)$ at every timestep $0.1s$.
	%
	%
	%
	%
	We then use Algorithm $1$ to synthesize a model and RPI sets required for RMPC synthesis. To this end, we follow the initialization procedure described in Section \ref{sec:main_section}-$4(b)$ to obtain an initial feasible $Z_{\mathrm{NL}}$ for Problem \eqref{eq:main_SDP}. We first parametrize the disturbance set $\mathcal{W}$ with $m_w=10$ hyperplanes. Then, following Step $(i)$, we characterize the set $\hat{\Sigma}_T$ with $\hat{\theta}_T=1\cdot10^{-3}$. Then, we compute the initial model $A={\scalemath{0.9}{\begin{bmatrix}0.9967 &0.0951 \\
						-0.0637  & 0.9036\end{bmatrix}}}$, $B={\scalemath{0.9}{\begin{bmatrix}0.0098 \\
						0.1914\end{bmatrix}}}$ and $\norm{d}_1=0.5816$ following Step ($ii$). We initialize the feedback gain as $K=[-0.4140 \ -2.3734]$ which is the optimal LQR gain corresponding to matrices $\tilde{Q}=\mathrm{diag}(1,15)$ and $\tilde{R}=1$. Then, we compute an RPI set $\Delta \mathcal{X}=\mathcal(\munderbar{P},\munderbar{b})$ following Step ($iii$) with $\munderbar{m}=10$ hyperplanes using \cite{Trodden2016}. Similarly, we compute the PI terminal set $\mathcal{X}_{\mathrm{t}}=\mathcal{P}(\bar{P},\bar{b})$ following Step ($iv$) with $\bar{m}=15$ hyperplanes using \cite{Gilbert1991}. Finally, with Step ($v$) we compute the remaining variables formulating $Z_{\mathrm{NL}}$. We parameterize $\mathcal{B}(\bar{\epsilon})$ with $m_{\epsilon}=10$ for terminal set maximization.
	The results obtained with Algorithm~1 with weights $\alpha=1,\beta=1,\gamma=0.1$ using the MOSEK SDP solver \cite{mosek} in MATLAB are shown in Figure~\ref{fig:main_result}. For the purpose of comparison, we also plot the results when the model $(A,B,d)$ is fixed to the initial value.
	%
	%
	We observe that by allowing Algorithm~1 to adapt the system model using $\hat{\Sigma}_T$, we obtain a lower objective value with a larger terminal set $\mathcal{X}_{\mathrm{t}}$ and a smaller RPI set $\Delta \mathcal{X}$.
	The model at termination is $A=\scalemath{0.95}{\begin{bmatrix}  0.9967  &  0.0951 \\
				-0.0625  &  0.8990 \end{bmatrix}}$,  $B=\scalemath{0.95}{\begin{bmatrix}  0.0098 \\
				0.1958 \end{bmatrix}}$, and ${\norm{d}_1 = 0.5833}$, and the computed feedback gain is $K=[-1.7062 \ -2.6306]$: the model consists of a larger disturbance set than the initialized value, with $(A,B,d)$ optimizing \eqref{eq:main_SDP} instead of best fitting the data.
	In case the model is fixed to the initial value, the feedback gain at termination is $K=[-1.0033   \ -3.0882]$. In order to study the effect of the parameter $\hat{\theta}_T$ characterizing $\hat{\Sigma}_T$, we run Algorithm 1 for increasing values of $\hat{\theta}_T$.  The objective values at termination are $11.631$ for $\hat{\theta}_T=1\cdot10^{-3}$, $11.659$ for $\hat{\theta}_T=1.2\cdot10^{-3}$,  $11.668$ for $\hat{\theta}_T=1.3\cdot10^{-3}$, $13.4376$ for $\hat{\theta}_T=1.5\cdot10^{-3}$: we observe that conservatism increases with $\hat{\theta}_T$, while increasing robustness with respect to the underlying plant. Note that this trend is not guaranteed since Problem \eqref{eq:main_SDP} is an NLPMI.
	
		\textit{Computational Complexity: }The SDP in \eqref{eq:update_SDP} consists of an LMI constraint with $\munderbar{m}(2n_x+m_w+n_u+\munderbar{m}+1)+\bar{m}(n_x+n_u+\bar{m}+1)+(m_x+m_u)(2n_x+\munderbar{m}+\bar{m}+1)+(3n_x+2n_u)+(n_x+\bar{m}+1)+n_x=1086$ rows, $n_xm_x^{{v}}=8$ linear equality constraints, and $2m_wT+2m_w(2n_x+n_u)+(2n_x+n_u)+m_w+2m_x^{{v}}(\bar{m}+m_{\epsilon})=20275$ linear inequality constraints over $2(n_x^2+n_xn_u+n_xm_x^v+1)+m_w(2n_x+n_u+1)+\munderbar{m}(\munderbar{m}+m_w+1+m_x+m_u)+\bar{m}(\bar{m}+2+m_x+m_u)+m_{\epsilon}
		%
		=638$ variables. Over multiple runs, the average solving time for the SDP in \eqref{eq:update_SDP} on a laptop with an Intel i7-7500U processor and 16GB of RAM running Ubuntu 16.04 is approximately $1.57$s when the model is allowed to adapt, and $0.78$s when the model is fixed. We note that the number of LMI constraints and variables scale quadratically in $\bar{m}$ and $\munderbar{m}$. Hence, the approach can be computationally expensive if a large number of hyperplanes are required for RPI set representation.
		%
		Comparing our approach to \cite{Yuxiao2021} using data from a real-world system is a subject of future work.
	\section{Conclusions}
	This paper has presented a data-driven method based on RPI sets to synthesize RMPC controllers. 
	 To this end, a set of LTI models that can describe the underlying plant behavior is characterized using an input-state dataset. Then, a suitable model along with RPI sets are concurrently computed for RMPC synthesis. This procedure is demonstrated to compute RPI sets with reduced conservativeness when compared to a sequential procedure.
	%
	%
	Future research will be devoted to estimation techniques for $\hat{\theta}_T$, reducing conservativeness in Theorem \ref{thm:robust_model_theorem}, using input-output datasets, multiplicative uncertainty models, and combining our approach with \cite{Berberich2020}.

	\bibliography{references}
\end{document}